\documentclass[numberwithinsect]{lipics-v2021}
\hideLIPIcs
\usepackage[T1]{fontenc}
\usepackage{microtype} 
\usepackage{hyperref}
\hypersetup{colorlinks=true,linkcolor=[rgb]{0.75,0,0},citecolor=[rgb]{0,0,0.75}}
\usepackage{graphicx}
\graphicspath{{./graphics/}}
\usepackage{amsmath}
\usepackage{amssymb}
\usepackage{pbox}
\usepackage{color}
\usepackage[table]{xcolor}
\newcommand{\graphicsScale}{.85}
\nolinenumbers
\newcommand{\para}[1]{\textbf{#1}~}
\newcommand{\OO}{\mathcal{O}} 
\newcommand{\floor}[1]{\lfloor #1 \rfloor}
\newcommand{\ceil}[1]{\lceil #1 \rceil}

\title{Short Flip Sequences to\texorpdfstring{\\}{} Untangle Segments in the Plane}
\funding{This work is supported by the French ANR PRC grant ADDS (ANR-19-CE48-0005).}
\titlerunning{Short Flip Sequences to Untangle Segments in the Plane}

\author{Guilherme D. da Fonseca}
{Aix-Marseille Université and LIS, France}
{guilherme.fonseca@lis-lab.fr}
{https://orcid.org/0000-0002-9807-028X}
{}
\author{Yan Gerard}
{Université Clermont Auvergne and LIMOS, France}
{yan.gerard@uca.fr}
{https://orcid.org/0000-0002-2664-0650}
{}
\author{Bastien Rivier}
{Université Clermont Auvergne and LIMOS, France}
{bastien.rivier@uca.fr}
{https://orcid.org/0000-0001-5985-2169}
{}

\authorrunning{G. D. da Fonseca, Y. Gerard, and B. Rivier}

\Copyright{Guilherme D. da Fonseca, Yan Gerard, and Bastien Rivier} 



\begin{document}
\maketitle

\begin{abstract}
A (multi)set of segments in the plane may form a TSP tour, a matching, a tree, or any multigraph. If two segments cross, then we can reduce the total length with the following \emph{flip} operation.
We \emph{remove} a pair of crossing segments, and \emph{insert} a pair of non-crossing segments, while keeping the same vertex degrees.
The goal of this paper is to devise efficient strategies to flip the segments in order to obtain crossing-free segments after a small number of flips.
Linear and near-linear bounds on the number of flips were only known for segments with endpoints in convex position. We generalize these results, proving linear and near-linear bounds for cases with endpoints that are not in convex position.
Our results are proved in a general setting that applies to multiple problems, using multigraphs and the distinction between removal and insertion choices when performing a flip.
\end{abstract}

\section{Introduction}
\label{sec:intro}

The\emph{Euclidean Travelling Salesman Problem (TSP)} is one of the most studied geometric optimization problems. We are given a set $P$ of points in the plane and the goal is to find a tour $S$ of minimum length. While the optimal solution has no crossing segments, essentially all approximation algorithms, heuristics, and PTASs may produce solutions $S$ with crossings. Given $S$, the only procedure known to obtain a solution $S'$ without crossings and of shorter length is to perform a flip operation. In our case, a \emph{flip} consists of \emph{removing} a pair of crossing segments, and then \emph{inserting} a pair of non-crossing segments preserving a tour (and consequently reducing its length). Flips are performed in \emph{sequence} until a crossing-free tour is obtained, in a procedure called \emph{untangle}.

The same flip operation may be applied in other settings. More precisely, a \emph{flip} consists of removing a pair of crossing segments $s_1,s_2$ and inserting a pair of segments $s'_1, s'_2$ in a way that $s_1,s'_1,s_2,s'_2$ forms a cycle and a certain graph \emph{property} is preserved. In the case of TSP tours, the property is being a Hamiltonian cycle. Other properties have also been studied, such as spanning trees, perfect matchings, and multigraphs. Notice that flips preserve the degrees of all vertices and multiple copies of the same edge may appear when we perform a flip on certain graphs.

When the goal is to obtain a crossing-free TSP tour, we are allowed to \emph{choose} which pair of crossing segments to remove in order to perform fewer flips, which we call \emph{removal choice} (Figure~\ref{fig:sequence}(a)). Notice that, in a tour, choosing which pair of crossing edges we remove defines which pair of crossing edges we insert. However, this is not the case for matchings and multigraphs. There, we are also allowed to choose which pair of segments to insert among two possibilities, which we call \emph{insertion choice} (Figure~\ref{fig:sequence}(b)).

\begin{figure}[htb]
 \centering
 \hfill
 \includegraphics[scale=\graphicsScale,page=1]{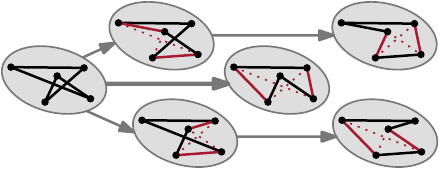} \hfill\hfill
 \includegraphics[scale=\graphicsScale,page=2]{sequence2}\hfill\phantom{.}\\
 \hfill(a)\hfill\hfill(b)\hfill\phantom{.}
 \caption{(a) Three untangle sequences for a tour with different \emph{removal choices}. (b) Three untangle sequences for a matching with different \emph{insertion choices}. We highlight the segments removed and inserted at each flip.}
 \label{fig:sequence}
\end{figure}

Using removal or insertion choices to obtain shorter flip sequences has not been explicitly studied before and opens several new questions, while unifying the solution to multiple reconfiguration problems. Next, we describe previous work according to which choices are used. Throughout, $P$ denotes the set of points and $n$ the number of segments.

\para{Using no choice:} Van Leeuwen et al.~\cite{VLSC81} showed that the \emph{length} (i.e. the number of flips) of any untangle sequence for a TSP tour is $\OO(n^3)$ and it is easy to construct $\Omega(n^2)$ examples. The same proof has been rediscovered in the context of matchings~\cite{BoM16} after 35 years. If $P$ is in convex position, then the number of crossings decreases at each flip, which gives a tight bound of $\Theta(n^2)$. If all points except the endpoints of $t$ segments are in convex position, then the authors~\cite{DGR23} recently showed a bound of $\OO(tn^2)$.

\para{Using only insertion choice:} Bonnet et al.~\cite{BoM16} showed that using only insertion choice, it is possible to untangle a matching using $\OO(n^2)$ flips. Let $\sigma$ be the \emph{spread} of $P$, that is, the ratio between the maximum and minimum distances among points in $P$. Using insertion choice, it is also possible to untangle a matching using $\OO(n \sigma)$ flips~\cite{BMS19}.

\para{Using only removal choice:} If $P$ is in convex position, then by using $\OO(n)$ flips we can untangle a TSP tour~\cite{OdW07,WCL09}, as well as a red-blue matching~\cite{BMS19}, while the best known bound for trees is $\OO(n \log n)$~\cite{BMS19}.
If instead of convex position, we have colinear red points in a red-blue matching, then $\OO(n^2)$ flips suffice~\cite{BMS19,DDFGR22}.

\para{Using both removal and insertion choices:} If $P$ is in convex position, then by using $\OO(n)$ flips we can untangle a matching~\cite{BMS19}.

\subsection{New Results}

Previous results are usually stated for a single graph property. Using choices, we are able to state the results in a more general setting. Proofs that use insertion choice are unlikely to generalize to red-blue matchings, TSP tours, or trees, where insertion choice is not available (still, they may hold for both non-bipartite matchings and multigraphs). In contrast, bounds for multigraphs using only removal choice apply to all these cases.
Previously, we only knew linear or near-linear bounds when the points $P$ are in convex position and removal choice is available.
The goal of the paper is to obtain linear and near-linear bounds to as many cases as possible, considering near-convex configurations as well as removal and insertion choices. 

Let $P = C \cup T$ where $C$ is in convex position and the points of $T$ are outside the convex hull of $C$, unless otherwise specified. Let $S$ be a multiset of $n$ segments with endpoints $P$ and $t$ be the number of segments with at least one endpoint in $T$. We prove the following results to untangle $S$, and some are summarized in Table~\ref{tab:results}.

\para{Using only insertion choice (Section~\ref{sec:Insertion}):} If $T=\emptyset$, then $\OO(n \log n)$ flips suffice. If $T$ is separated from $C$ by two parallel lines, then $\OO(t n \log n)$ flips suffice.

\para{Using only removal choice (Section~\ref{sec:Removal}):} If $|T| \leq 2$ and $t = \OO(1)$, then $\OO(n \log n)$ flips suffice. In this case, our results hold with the points $T$ being anywhere with respect to the convex hull of $C$. As the bounds hold for trees, it is useful to compare them against the $\OO(n \log n)$ bound for trees from~\cite{BMS19} that strongly uses the fact that $S$ forms a tree.
The $\OO(\log n)$ factor is not present for the special cases of TSP tours and red-blue matchings.

\para{Using both removal and insertion choices (Section~\ref{sec:RemovalInsertion}):} If $T$ is separated from $C$ by two parallel lines, then $\OO(t n)$ flips suffice. If $T$ is anywhere outside the convex hull of $C$ and $S$ is a matching, then $\OO(t^3 n)$ flips suffice.

\begin{table}[htb]
    \centering
    \caption{Upper bounds to different versions of the problem with points having $\OO(1)$ degree. The letter R corresponds to removal choice, I to insertion choice, and $\emptyset$ to no choice. New results are highlighted in yellow with the theorem number in parenthesis and tight bounds are bold.}
    \label{tab:results}
    \newcommand{\mc}[2]{\multicolumn{#1}{l|}{#2}}
    \begin{tabular}{|l||l|l|l|l|l|l|}\cline{1-7}
    Property: & \mc{4}{Matching} & \mc{2}{TSP, Red-Blue}\\\hline
    Choices: & RI & I & R & $\emptyset$ & R & $\emptyset$\\\hline\hline
    Convex & $\mathbf{n}$~\cite{BMS19} &  \cellcolor{yellow!50}$n \log n$ (\ref{thm:convexI}) & \cellcolor{yellow!50}$n \log n$ (\ref{thm:convexR})& $\mathbf{n^2}$ & $\mathbf{n}$ \cite{BMS19,OdW07,WCL09} & $\mathbf{n^2}$ \\\hline
    $|T| = 1$ & \cellcolor{yellow!50}$\mathbf{n}$ (\ref{thm:separatedRI}) & \cellcolor{yellow!50}$n \log n$ (\ref{thm:separatedI}) & \cellcolor{yellow!50}$n \log n$ (\ref{thm:1InsideOutsideR}) & $\mathbf{n^2}$~\cite{DGR23} & \cellcolor{yellow!50}$\mathbf{n}$  (\ref{thm:1InsideOutsideR}) & $\mathbf{n^2}$~\cite{DGR23} \\\hline
    $|T| = 2$ & \cellcolor{yellow!50}$\mathbf{n}$ (\ref{thm:nearConvexRI}) & $n^2$~\cite{BoM16} & \cellcolor{yellow!50}$n \log n$  (\ref{thm:2OutsideR}) & $\mathbf{n^2}$~\cite{DGR23} & \cellcolor{yellow!50}$\mathbf{n}$  (\ref{thm:2OutsideR}) & $\mathbf{n^2}$~\cite{DGR23} \\\hline
    separated & \cellcolor{yellow!50}$tn$ (\ref{thm:separatedRI}) & \cellcolor{yellow!50}$tn \log n$ (\ref{thm:separatedI})& \mc{4}{$tn^2$~\cite{DGR23}} \\\hline
    $C \cup T$ & \cellcolor{yellow!50}$t^3n$ (\ref{thm:nearConvexRI})& $n^2$~\cite{BoM16} & \mc{4}{$tn^2$~\cite{DGR23}} \\\hline
    
    \end{tabular}
\end{table}

In a matching or TSP tour, we have $t = \OO(|T|)$ and $n = \OO(|P|)$, however in a tree, $t$ can be as high as $\OO(|T|^2)$. In a multigraph $t$ and $n$ can be much larger than $|T|$ and $|P|$. The theorems describe more precise bounds as functions of all these parameters. For simplicity, the introduction only shows bounds in terms of only $n$ and $t$.

\subsection{Related Reconfiguration Problems}

Combinatorial reconfiguration studies the step-by-step transition from one solution to another, for a given combinatorial problem. Many reconfiguration problems are presented in~\cite{Heu13}. We give a brief overlook of reconfiguration among line segments using alternative flip operations.

The \emph{2OPT flip} is not restricted to crossing segments. It removes and inserts pairs of segments (the four segments forming a cycle) as the total length decreases. In contrast to flips among crossing segments, the number of 2OPT flips performed may be exponential~\cite{ERV14}. 

It is possible to relax the flip definition even further to all operations that replace two segments by two others forming a cycle~\cite{BeI08,BeI17,BBH19,BJ20,EKM13,Wil99}. This definition has also been considered for multigraphs~\cite{Hak62,Hak63,Joff20}. 

Another type of flip consists of removing a single segment and inserting another one. 
Such flips are widely studied for triangulations~\cite{AMP15,HNU99,Law72,LuP15,NiN18,Pil14}.
They have also been considered for non-crossing trees~\cite{ABDK22} and paths. It is possible to reconfigure any two non-crossing paths if the points are in convex position~\cite{AIM07,ChWu09} or if there is \emph{one} point inside the convex hull~\cite{AKLM23}.

\subsection{Preliminaries}
\label{sec:preliminariesBackup}

Throughout, we consider \emph{multigraphs} $(P,S)$ whose vertices $P$ (called \emph{endpoints}) are points in the plane and edges $S$ are a multiset of line \emph{segments}.
We assume that the endpoints are in \emph{general position} and that the two endpoints of a segment are distinct.
Given two (possibly equal) sets $P_1,P_2$ of endpoints, we say that a segment is a \emph{$P_1P_2$-segment} if one endpoint is in $P_1$ and the other is in $P_2$. Similarly, we say that a segment is a \emph{$P_1$-segment} if at least one endpoint is in $P_1$. 

We say that two segments \emph{cross} if they intersect at a single point that is not an endpoint of either segment. We say that a line \emph{crosses} a segment if they intersect at a single point that is not an endpoint of the segment.
We say that a segment or a line $h$ \emph{separates} a set of points $P$ if $P$ can be partitioned into two non-empty sets $P_1,P_2$ such that every segment $p_1p_2$ with $p_1 \in P_1, p_2 \in P_2$ crosses $h$.
Given a set of segments $S$, the \emph{line potential} $\lambda(\ell)$ is the number of segments of $S$ crossed by $\ell$.
Several proofs in this paper use the following two lemmas from previous papers.

\begin{lemma}[\cite{VLSC81}]
  \label{lem:lambda}
  Given a multiset $S$ of segments and a line $\ell$, let $\lambda(\ell)$ be the number of segments in $S$ crossing $\ell$. Then, $\lambda(\ell)$ never increases at a flip.
\end{lemma}

\begin{lemma}[\cite{BoM16}]
  \label{lem:split}
Consider a partition $S=\bigcup_i S_i$ of the multiset $S$ of segments and let $P_i$ be the set of endpoints of $S_i$. If no segment of $\binom{P_i}{2}$ crosses a segment of $\binom{P_j}{2}$ for $i \neq j$, then the sequences of flips in each $S_i$ are independent.
\end{lemma}

We say that a segment $s$ is \emph{uncrossable} if for any two endpoints $p_1,p_2$, we have that $p_1p_2$ do not cross $s$. Lemma~\ref{lem:split} implies that an uncrossable segment cannot be flipped.

Our bounds often have terms like $\OO(tn)$ and $\OO(n \log |C|)$ that would incorrectly become $0$ if $t$ or $\log |C|$ is $0$. In order to avoid this problem, factors in the $\OO$ notation should be made at least 1. For example, the aforementioned bounds should be respectively interpreted as $\OO((1+t) n)$ and $\OO(n \log (2+|C|))$.

\subsection{Techniques}

To prove our results, we combine previous and new \emph{potential functions} with refined strategies and analysis. Van Leeuwen et al.~\cite{VLSC81} as well as Bonnet et al.~\cite{BoM16} consider $\lambda(\ell)$ for a set $L$ of all $\OO(|P|^2)$ lines defined by $P$. Since there always exists a line in $L$ whose potential decreases at a flip, we obtain the $\OO(|P|^2n) = \OO(n^3)$ classical bound without any choice.
Bonnet et al.~\cite{BoM16} show that a set $L$ of $|P|-1$ parallel lines (with one point between two consecutive lines) suffices with insertion choice. Since there is always an insertion choice that makes some $\lambda(\ell)$ decrease for $\ell \in L$, the $\OO(|P|n) = \OO(n^2)$ bound follows.

In order to avoid a quadratic dependency in $n$, new line potentials have to be introduced with careful removal and/or insertion choices.
For example, to prove Theorem~\ref{thm:nearConvexRI}, we have to perform several flips in order to find a line $\ell$ with $\lambda(\ell) = \OO(t)$ before applying the line potential argument. In contrast, to prove Theorem~\ref{thm:separatedI}, we modify the line potential to only count the $t$ $T$-segments. However, with this change the line potential may increase, which we need to handle properly.

Another key potential, inspired by~\cite{BMS19} and used for the convex case is the \emph{depth potential} $\delta(p_ap_b)$ of a segment $p_ap_b$, defined as the number of points between $p_a$ and $p_b$ along the convex hull boundary with a given orientation. 
Careful removal and insertion choices as well as adaptations of this potential had to be made in order to guarantee that the potential decreases during most flips and never increases by too much. For example, to prove Theorems~\ref{thm:convexI} and~\ref{thm:separatedI}, we had to consider the product of the depth, instead of the usual sum. To prove Theorem~\ref{thm:convexR}, we had to modify the depth potential to only count endpoints of segments that have crossings, which we call the \emph{crossing depth} $\delta_{\times}(p_ap_b)$.

In the convex case, the number of crossings decreases at each flip, which implies the trivial $\binom{n}{2}$ upper bound. However, the number of crossings may increase when the points are not in convex position. An analysis of the number of crossings is used to bound the number of flips in the proof of Theorem~\ref{thm:2InsideR}.

Finally, we use the concept of \emph{splitting} from~\cite{BoM16}, presented in Lemma~\ref{lem:split}. The difficulty of splitting is to obtain the disjoint sets required by the lemma. For example, in Theorem~\ref{thm:nearConvexRI}, we untangle segments with both endpoints in $C$ last to obtain the desired separation. In Theorem~\ref{thm:2OutsideR}, we carefully find lines that split the original problem into problems with a smaller value of $t$ that are solved recursively. The special case of uncrossable segments is used in Theorems~\ref{thm:1InsideOutsideR} and \ref{thm:2OutsideR}.

\section{Insertion Choice}
\label{sec:Insertion}

In this section, we show how to untangle a multigraph using only insertion choice, that is, our strategies do not choose which pair of crossing segments is removed, but only which pair of segments with the same endpoints is subsequently inserted. We start with the convex case, followed by points outside the convex separated by two parallel lines.

\subsection{Convex}

Let $P = C = \{p_1,\ldots,p_{|C|}\}$ be a set of points in convex position sorted in counterclockwise order along the convex hull boundary (Figure~\ref{fig:convexI}(a)). Given a segment $p_ap_b$, we define the \emph{depth} $\delta(p_ap_b) = |b-a|$. This definition resembles but is not the same as the depth used in~\cite{BMS19}. We use the depth to prove the following theorem.

\begin{figure}[htb]
 \centering
 \hfill
 \includegraphics[scale=\graphicsScale,page=1]{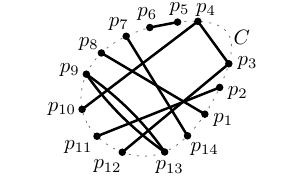}\hfill\hfill
 \includegraphics[scale=\graphicsScale,page=2]{convexI}\hfill\hfill
 \includegraphics[scale=\graphicsScale,page=3]{convexI}\hfill\phantom{.}\\
 \hfill(a)\hfill\hfill(b)\hfill\hfill(c)\hfill\phantom{.}
 \caption{(a) A multigraph $(C,S)$ with $|C|=14$ points in convex position and $n=9$ segments. (b) Insertion choice for Case~1 and 2 of the proof of Theorem~\ref{thm:convexI}. (c) Insertion choice for Case~3.}
 \label{fig:convexI}
\end{figure}

\begin{theorem} \label{thm:convexI}
Every multigraph $(C,S)$ with $C$ in convex position has an untangle sequence of length $\OO(n \log |C|) = \OO(n \log n)$ using only insertion choice, where $n = |S|$.
\end{theorem}
\begin{proof}
Let the potential function
\[\phi(S) = \prod_{s \in S} \delta(s).\]
As $\delta(s) \in \{1,\ldots,|C|-1\}$, we have that $\phi(S)$ is integer, positive, and at most $|C|^n$. Next, we show that for any flipped pair of segments $p_ap_b,p_cp_d$ there exists an insertion choice that multiplies $\phi(S)$ by a factor of at most $3/4$, and the theorem follows.

Consider a flip of a segment $p_ap_b$ with a segment $p_cp_d$ and assume without loss of generality that $a < c < b < d$.
The contribution of the pair of segments $p_ap_b,p_cp_d$ to the potential $\phi(S)$ is the factor $f=\delta(p_ap_b)\delta(p_cp_d)$.
Let $f'$ be the factor corresponding to the pair of inserted segments.

\textbf{Case~1:} If $\delta(p_ap_c) \leq \delta(p_cp_b)$, then we insert the segments $p_ap_c$ and $p_bp_d$ and we get $f'=\delta(p_ap_c)\delta(p_bp_d)$ (Figure~\ref{fig:convexI}(b)).
We notice $\delta(p_ap_b)=\delta(p_ap_c)+\delta(p_cp_b)$. It follows $\delta(p_ap_c) \leq \delta(p_ap_b)/2$ and we have $\delta(p_bp_d) \leq \delta(p_cp_d)$ and then $f'\leq f/2$. 

\textbf{Case~2:} If $\delta(p_bp_d) \leq \delta(p_cp_b)$, then we insert the same segments $p_ap_c$ and $p_bp_d$ as previously. We have $\delta(p_ap_c) \leq \delta(p_ap_b)$ and $\delta(p_bp_d)\leq \delta(p_cp_d)/2$, which gives $f'\leq f/2$.

\textbf{Case~3:} If (i) $\delta(p_ap_c) > \delta(p_cp_b)$ and (ii) $\delta(p_bp_d) > \delta(p_cp_b)$, then we insert the segments $p_ap_d$ and $p_cp_b$ (Figure~\ref{fig:convexI}(c)). 
The contribution of the new pair of segments is $f'=\delta(p_ap_d)\delta(p_cp_b)$.
We introduce the coefficients $x=\frac{\delta(p_ap_c)}{\delta(p_cp_b)}$ and $y=\frac{\delta(p_bp_d)}{\delta(p_cp_b)}$ so that $\delta(p_ap_c) = x\delta(p_cp_b)$ and $\delta(p_bp_d) = y\delta(p_cp_b)$. It follows that $\delta(p_ap_b) = (1+x)\delta(p_cp_b)$, $\delta(p_cp_d)=(1+y)\delta(p_cp_b)$ and $\delta(p_ap_d) = (1+x+y)\delta(p_cp_b)$. The ratio $f'/f$ is equal to a function $g(x,y) = \frac{1+x+y}{(1+x)(1+y)}$. Due to (i) and (ii), we have that $x\geq 1$ and $y \geq 1$. 
In other words, we can upper bound the ratio $f'/f$ by the maximum of the function $g(x,y)$ with $x,y \geq 1$. It is easy to show that the function $g(x,y)$ is decreasing with both $x$ and $y$. Then its maximum 
is obtained for $x=y=1$ and it is equal to $3/4$, showing that $f'\leq 3f/4$.
\end{proof}

\subsection{Separated by Two Parallel Lines}

In this section, we prove the following theorem, which is a generalization of Theorem~\ref{thm:convexI}.

\begin{figure}[htb]
 \centering
 \hfill
 \includegraphics[scale=\graphicsScale,page=1]{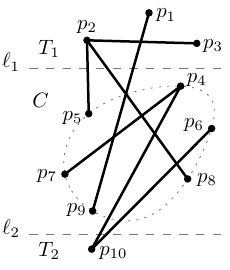} \hfill\hfill
 \includegraphics[scale=\graphicsScale,page=2]{separatedI} \hfill\phantom{.}\\
 \hfill(a)\hfill\hfill(b)\hfill\hfill\phantom{.}
 \caption{(a) Statement of Theorem~\ref{thm:separatedI}. (b) Some insertion choices in the proof of Theorem~\ref{thm:separatedI}.}
 \label{fig:separatedI}
\end{figure}

\begin{theorem} \label{thm:separatedI}
Consider a multigraph $(P,S)$ with $P = C \cup T_1 \cup T_2$ where $C$ is in convex position and there exist two horizontal lines $\ell_1,\ell_2$, with $T_1$ above $\ell_1$ above $C$ above $\ell_2$ above $T_2$.
Let $T = T_1 \cup T_2$, $n = |S|$, and $t$ be the number of $T$-segments.
There exists an untangle sequence of length $\OO(t |P| \log |C| + n \log |C|) = \OO(tn \log n)$ using only insertion choice.
\end{theorem}
\begin{proof}
We start by describing the insertion choice for flips involving at least one point in $T$.
Let $p_1,\ldots,p_{|P|}$ be the points $P$ sorted vertically from top to bottom.
Consider a flip involving the points $p_a,p_b,p_c,p_d$ with $a<b<c<d$. The insertion choice is to create the segments $p_ap_b$ and $p_cp_d$. See Figure~\ref{fig:separatedI}(b). As in~\cite{BoM16}, we define the potential $\eta$ of a segment $p_ip_j$ as
\[\eta(p_ip_j) = |i-j|.\]
Notice that $\eta$ is an integer between $1$ and $|P|-1$. We define $\eta_T(S)$ as the sum of $\eta(p_ip_j)$ for $p_ip_j \in S$ with $p_i$ or $p_j$ in $T$. Notice that $0 < \eta_T(S) < t |P|$. It is easy to verify that any flip involving a point in $T$ decreases $\eta_T(S)$ and other flips do not change $\eta_T(S)$. Hence, the number of flips involving at least one point in $T$ is $\OO(t|P|)$.

For the flips involving only points of $C$, we use the same choice as in the proof of Theorem~\ref{thm:convexI}.
The potential function 
 \[\phi(S) = \prod_{p_ip_j \in S \;:\; p_i\in C \text{ and } p_j \in C} \delta(p_ip_j)\]
is at most $|C|^n$ and decreases by a factor of at most $3/4$ at every flip that involves only points of $C$.

However, $\phi(S)$ may increase by a factor of $\OO(|C|^2)$ when performing a flip that involves a point in $T$. As such flips only happen $\OO(t|P|)$ times, the total increase is at most a factor of $|C|^{\OO(t|P|)}$.

Concluding, the number of flips involving only points in $C$ is at most 
\[\log_{4/3}\left(|C|^{\OO(n)} |C|^{\OO(t|P|)} \right) = \OO(n \log |C| + t|P| \log |C|).\qedhere\]
\end{proof}

\section{Removal Choice}
\label{sec:Removal}

In this section, we show how to untangle a multigraph using only removal choice. We start with the convex case, followed by $1$ point inside or outside the convex, then $2$ points outside the convex, $2$ points inside the convex, and $1$ point inside and $1$ outside the convex. As only removal choice is used, all results also apply to red-blue matchings, TSP tours, and trees.

\subsection{Convex}
\label{sec:convexR}

Let $P = C = \{p_1,\ldots,p_{|C|}\}$ be a set of points in convex position sorted in counterclockwise order along the convex hull boundary and consider a set of segments $S$ with endpoints $P$. Given a segment $p_ap_b$ and assuming without loss of generality that $a<b$, we define the \emph{crossing depth} $\delta_\times(p_ap_b)$ as the number of points in $p_{a+1},\ldots,p_{b-1}$ that are an endpoint of a segment in $S$ that crosses any other segment in $S$ (not necessarily $p_ap_b$). We use the crossing depth to prove the following theorem, which implies a simpler and more general proof of the $\OO(n \log n)$ bound for trees~\cite{BMS19}.

\begin{figure}[htb]
 \centering
 \hfill
 \includegraphics[scale=\graphicsScale,page=1]{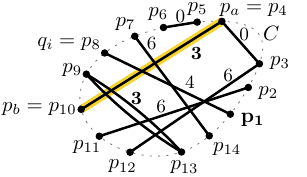} \hfill\hfill
 \includegraphics[scale=\graphicsScale,page=2]{convexR}~ \hfill\hfill
 \includegraphics[scale=\graphicsScale,page=3]{convexR}\hfill\phantom{.} \\
 \hfill(a)\hfill\hfill(b)\hfill\hfill(c)\hfill\phantom{.}
 \caption{Proof of Theorem~\ref{thm:convexR}. (a) The segments of a convex multigraph are labeled with the crossing depth. (b,c) Two possible pairs of inserted segments, with one segment of the pair having crossing depth $\floor{\frac{3}{2}}=1$.
 }
 \label{fig:convexR}
\end{figure}

\begin{theorem} \label{thm:convexR}
Every multigraph $(C,S)$ with $C$ in convex position has an untangle sequence of length $\OO(n \log |C|) = \OO(n \log n)$ using only removal choice, where $n = |S|$.
\end{theorem}
\begin{proof}
We repeat the following procedure until there are no more crossings. Let $p_ap_b \in S$ be a segment \emph{with crossings} (hence, crossing depth at least one) and $a<b$ minimizing $\delta_\times(p_ap_b)$ (Figure~\ref{fig:convexR}(a)). Let $q_1,\ldots,q_{\delta_\times(p_ap_b)}$ be the points defining $\delta_\times(p_ap_b)$ in order and let $i = \ceil{\delta_\times(p_ap_b)/2}$. Since $p_ap_b$ has minimum crossing depth, the point $q_i$ is the endpoint of segment $q_ip_c$ that crosses $p_ap_b$. When flipping $q_ip_c$ and $p_ap_b$, we obtain a segment $s$ (either $s=q_ip_a$ or $s=q_ip_b$) with $\delta_\times(s)$ at most half of the original value of $\delta_\times(p_ap_b)$ (Figure~\ref{fig:convexR}(b,c)). Hence, this operation always divides the value of the smallest positive crossing depth by at least two. As the crossing depth is an integer smaller than $|C|$, after performing this operation $\OO(\log |C|)$ times, it produces a segment of crossing depth $0$. As the segments of crossing depth $0$ can no longer participate in a flip, the claimed bound follows.
\end{proof}

\subsection{One Point Inside or Outside a Convex}
\label{sec:1InsideOutsideR}

In this section, we prove Theorem~\ref{thm:1InsideOutsideR}.
In the case of TSP tours~\cite{OdW07,WCL09} and red-blue matchings~\cite{BBH19}, the preprocessing to untangle $CC$-segments takes $\OO(n)$ flips. However, in the case of trees~\cite{BBH19} and in general (Theorem~\ref{thm:convexR}), the best bound known is $\OO(n \log n)$.
We first state a lemma used to prove Theorem~\ref{thm:1InsideOutsideR}.

\begin{lemma}
    \label{lem:farthestFirst}
    Consider a set $C$ of points in convex position, and a multiset $S$ of $n$ crossing-free segments with endpoints in $C$. Consider the multiset $S \cup \{s\}$ where $s$ is an extra segment with one endpoint in $C$ and one endpoint $q$ anywhere in the plane. 
    There exists an untangle sequence for $S \cup \{s\}$ of length $\OO(n)$ using only removal choice.
\end{lemma}

\begin{proof}
    Iteratively flip the segment $qp_1$ with the segment $p_2p_3 \in S$ crossing $qp_1$ the farthest from $q$.
    This flip inserts a $CC$-segment $p_1p_2$, which is impossible to flip again, because the line $p_1p_2$ is crossing free. The flip does not create any crossing between $CC$-segments.
\end{proof}

We are now ready to state and prove the theorem.

\begin{theorem} \label{thm:1InsideOutsideR}
Consider a multigraph $(P,S)$ with $P = C \cup T$ where $C$ is in convex position, and $T = \{q\}$, and such that there is no crossing pair of $CC$-segments (possibly after a preprocessing for the convex case).
Let $n = |S|$ and $t$ be the number of $T$-segments.
There exists an untangle sequence of length $\OO(tn)$ using only removal choice.
\end{theorem}
\begin{proof}
For each segment $s$ with endpoint $q$ with crossing, we apply Lemma~\ref{lem:farthestFirst} to $s$ and the $CC$-segments crossing $s$.
Once a segment $s$ incident to $q$ is crossing free, it is impossible to flip it again as we fall in one of the following cases.
Let $\ell$ be the line containing $s$.

\textbf{Case~1:} If $\ell$ is crossing free, then it splits the multigraph in three partitions: the segments on one side of $\ell$, the segments on the other side of $\ell$, and the segment $s$ itself.

\textbf{Case~2:} If $\ell$ is not crossing free and $q$ is outside the convex hull of $C$, then $s$ is \emph{uncrossable}.

\textbf{Case~3:} If $q$ is inside the convex hull of $C$, then introducing a crossing on $s$ would require that $q$ lies in the interior of the convex quadrilateral whose diagonals are the two segments removed by a flip. The procedure excludes this possibility by ensuring that there are no crossing pair of $CC$-segments, and, therefore, that one of the removed segment already has $q$ as an endpoint.

Therefore, we need at most $n$ flips for each of the $t$ segments incident to $q$.
\end{proof}

\subsection{Two Points Outside a Convex}
\label{sec:2OutsideR}

In this section, we prove a theorem with a bound that is exponential in $t$, which makes it of little interest for large $t$. Notice, however, that in matchings $t \leq 2$, in a TSP tour $t \leq 4$, and in a binary tree $t \leq 6$. Also notice that the definition of $t$ is different from other theorems (here $TT$-segments are counted twice). Both definitions are equivalent up to a factor of $2$, but since $t$ appears in the exponent, they are not exchangeable.

\begin{theorem} \label{thm:2OutsideR}
Consider a multigraph $(P,S)$ with $P = C \cup T$ where $C$ is in convex position, the points of $T$ are outside the convex hull of $C$, and $|T| \leq 2$. 
Let $n = |S|$ and $t$ be the sum of the degrees of the points in $T$.
There exists an untangle sequence of length $\OO(2^t d_{\text{conv}}(n))$ using only removal choice, where $d_{\text{conv}}(n)$ is the number of flips to untangle any multiset of at most $n$ segments with endpoints in convex position.
\end{theorem}

\begin{proof}
Throughout this proof, we partition the $TT$-segments respectively the $CT$-segments into two types: \emph{$TTI$-segment} and \emph{$CTI$-segment} if it intersects the interior of the convex hull of $C$ and \emph{$TTO$-segment} and \emph{$CTO$-segment} otherwise.
Let $f(t)$ be the number of flips to untangle a multiset $S$ as in the statement of the theorem. The proof proceeds by induction. The base case is $t = 0$, when $f(0) \leq d_{\text{conv}}(n)$ by definition of $d_{\text{conv}}(n)$.

Next, we show how to bound $f(t)$ for $t > 0$, but first we need some definitions. A line $\ell$ is a \emph{$T$-splitter} if $\ell$ is crossing free and either $\ell$ contains a $T$-segment or there are $T$-segments on both sides of $\ell$. We abusively say that a segment $s$ is a $T$-splitter if the line containing $s$ is a $T$-splitter. A $T$-splitter is useful because we can apply Lemma~\ref{lem:split} and solve sub-problems with a lower value of $t$ by induction. 

\para{Phase~1: untangle all but one segment by induction.} We remove an arbitrary $CT$-segment or $TT$-segment $s$ from $S$. We then use induction to untangle $S$ using $f(t-1)$ flips and insert the segment $s$ back in $S$ afterwards. Notice that all crossings are now on $s$.

\para{Phase~2.1: apply induction if possible.}
If $S$ admits a $T$-splitter $\ell$, then we apply Lemma~\ref{lem:split} to solve each side of $\ell$ independently using induction.

If $S$ has a crossing-free $TTO$-segment $qq'$ such that the line $qq'$ is not crossing free, then $qq'$ is uncrossable, and we remove $qq'$ from $S$ and untangle $S$ by induction.
Similarly, in the case where $T=\{q,q'\}$ and where $qq'$ is a $TTI$-segment, if $S$ has a $CTO$-segment, say $pq$, then $pq$ is uncrossable, and we remove $pq$ from $S$ and untangle $S$ by induction.

In all the three cases of Phase~2.1 we get $f(t) \leq f(t-1) + f(t_1) + f(t_2)$, where $t_1+t_2 \leq t$ and $t_1,t_2 \geq 1$. 

\para{Phase~2.2: split after one flip.}
If $S$ contains no $T$-splitter and if $s$ is a $TT$-segment, then there remains no $CT$-segment in $S$ (as every $CT$-segment shares an endpoint with the $TT$-segment $s$ that contains all crossings), and $s$ crosses a $CC$-segment $s'$. 
A crossing-free $CT$-segment would either be a $CTI$-segment, hence a $T$-splitter, or a $CTO$-segment and, hence uncrossable and removed by one of the induction cases of Phase~2.1.

The segment $s'$ becomes a $T$-splitter after flipping $s$ with $s'$, and we invoke induction. 
By Lemma~\ref{lem:split}, we get in this case $f(t) \leq f(t-1) + 1 + f(t_1) + f(t_2)$, where $t_1+t_2 \leq t$ and $t_1,t_2 \geq 1$.

\para{Phase~2.3: split after $\OO(n)$ flips.}
In this case, $S$ contains no $T$-splitter and $s$ is a $CT$-segment, say with $q$ as its endpoint in $T$.
While $s'$, the segment of $S$ that crosses $s$ the farthest away from $q$, is a $CC$-segment, we flip $s$ and $s'$ and we set $s$ to be the newly inserted $CT$-segment incident to $q$. By Lemma~\ref{lem:farthestFirst}, at most $n$ flips are performed in this loop.

At the end of the loop, either $s$ is crossing free, or $s'$ is a $CT$-segment, say with $q'$ as its endpoint in $T$.
Then, we also flip $s$ and $s'$. 

\textbf{Insertion case~1:}
If two $CT$-segments are inserted, then, either one of them is uncrossable (this is the case if $s'$ is a $CTO$-segment), or $s'$ is now a $T$-splitter (recall that if 
$qq'$ is a $TTI$-segment, then all the $CTO$-segments have been removed at Phase~2.1.).

\textbf{Insertion case~2:}
If the $TT$-segment $qq'$ is inserted, then the inserted $CC$-segment is crossing free (as in the proof of Lemma~\ref{lem:farthestFirst}), and, if $qq'$ is not already crossing free, we flip $qq'$ with any segment, say $pp'$.

Next, we split $S$ as follows.
Among the $CTI$-segments of $S$ which are on the upper (respectively lower) side of the line $qq'$, consider the one whose endpoint $p_{\text{upper}}$ (respectively $p_{\text{lower}}$) in $C$ is the closest to the line $qq'$.
The segments of $S$ are either inside or outside the convex quadrilateral $qp_{\text{lower}}q'p_{\text{upper}}$, and we know that only the segments inside may have crossings.
By Lemma~\ref{lem:split}, we remove from $S$ all the segments outside $qp_{\text{lower}}q'p_{\text{upper}}$.
Recall that, in our case, $qq'$ is a $TTI$-segment, and all the $CTO$-segments have been removed at Phase~2.1.
The line $pp'$ is finally a $T$-splitter.
Again, by Lemma~\ref{lem:split}, we get in this case $f(t) \leq f(t-1) + n+2 + f(t_1) + f(t_2)$, where $t_1+t_2 \leq t$ and $t_1,t_2 \geq 1$.

\para{}
The last bound on $f(t)$ dominates the recurrence. Using that $f(t_1) + f(t_2) \leq f(t-1) + f(1)$ and $t<n$ we get
\[f(t) \leq f(t-1) + n+2 + f(t_1) + f(t_2) \leq \OO(n) + 2 f(t-1),\]
which solves to $f(t) = \OO(2^t d_{\text{conv}}(n))$ as claimed.
\end{proof}

\subsection{Two Points inside a Convex}
\label{sec:2insideR}

We prove a similar theorem for two points inside the convex hull of $C$.

\begin{theorem} \label{thm:2InsideR}
Consider a multigraph $(P,S)$ with $P = C \cup T$ where $C$ is in convex position, the points of $T$ are inside the convex hull of $C$, and $T = \{q,q'\}$. 
Let $n = |S|$ and $t$ be the number of $T$-segments.
There exists an untangle sequence of length $\OO(d_{\text{conv}}(n) + tn)$ using only removal choice, where $d_{\text{conv}}(n)$ is the number of flips to untangle any multiset of at most $n$ segments with endpoints in convex position.
\end{theorem}
\begin{proof}
The untangle sequence is decomposed in five phases. At the end of each phase, a new type of crossings is removed, and types of crossings removed in the previous phases are not present, even if they may temporarily appear during the phase.

\textbf{Phase~1 ($\mathbf{CT\times CT}$).} In this phase, we remove all crossings between pairs of $CT$-segments using $\OO(d_{\text{conv}}(t)) = \OO(d_{\text{conv}}(n))$ flips. We separately solve two convex sub-problems defined by the $CT$-segments, one on each side of the line $qq'$.

\textbf{Phase~2 ($\mathbf{CC\times CC}$).} In this phase, we remove all crossings between pairs of $CC$-segments using $\OO(d_{\text{conv}}(n))$ flips. As no $CT$-segment has been created, there is still no crossing between a pair of $CT$ segments. Throughout, our removal will preserve the invariant that no pair of $CC$-segments crosses.

\textbf{Phase~3 ($\mathbf{CT \times \textbf{non-central } CC}$).} 
We distinguish between a few types of $CC$-segments. The \textit{central} $CC$-segments cross the segment $qq'$ (regardless of $qq'$ being in $S$ or not), while the \emph{non-central} do not. The \textit{peripheral} $CC$-segments cross the line $qq'$ but not the segment $qq'$, while the \emph{outermost} $CC$-segments do not cross either. In this phase, we remove all crossings between $CT$-segments and non-central $CC$-segments.

Given a non-central $CC$-segment $pp'$, let the \emph{out-depth} $\delta'(pp')$ be the number of points of $C$ that are contained inside the halfplane bounded by the line $pp'$ and not containing $T$. Also, let $\chi$ be the number of crossings between the non-central $CC$-segments and the $CT$-segments. At the end of each step the two following invariants are preserved. (i) No pair of $CC$-segments crosses. (ii) No pair of $CT$-segments crosses.

At each step, we choose to flip the non-central $CC$-segment $pp'$ of minimum out-depth that crosses a $CT$-segment. We flip $pp'$ with the $CT$-segment $q''p''$ (with $q'' \in \{q,q'\}$)
that crosses $pp'$ at the point closest to $p$ (Figure~\ref{fig:2InsideR-3-1}(a) and Figure~\ref{fig:2InsideR-3-2}(a)).
One of the possibly inserted pairs may contain a $CT$-segment $s$ that crosses another $CT$-segment $s'$, violating the invariant (ii) (Figure~\ref{fig:2InsideR-3-1}(b) and Figure~\ref{fig:2InsideR-3-2}(b)). If there are multiple such segments $s'$, then we consider $s'$ to be the segment whose crossing with $s$ is closer to $q''$. We flip $s$ and $s'$ and obtain either two $CT$-segments (Figure~\ref{fig:2InsideR-3-1}(c) and Figure~\ref{fig:2InsideR-3-2}(c)) or a $CC$-segment and the segment $qq'$ (Figure~\ref{fig:2InsideR-3-1}(d) and Figure~\ref{fig:2InsideR-3-2}(d)). The analysis is divided in two main cases.

\begin{figure}[htb]
 \centering
 \phantom{.}\hfill\includegraphics[scale=\graphicsScale,page=2]{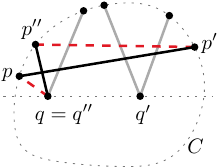}\hfill\hfill
 \includegraphics[scale=\graphicsScale,page=3]{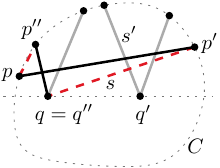}\hfill\hfill
 \includegraphics[scale=\graphicsScale,page=5]{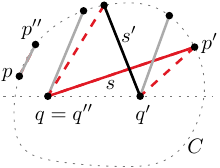}\hfill\hfill
 \includegraphics[scale=\graphicsScale,page=4]{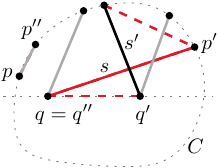}\hfill\phantom{.}\\
 \phantom{.}\hfill (a) \hfill\hfill (b) \hfill\hfill (c) \hfill\hfill (d) \hfill\phantom{.}
 \caption{Theorem~\ref{thm:2InsideR}, Phase~3 when $pp'$ is an outermost segment.}
 \label{fig:2InsideR-3-1}
\end{figure}

If $pp'$ is an outermost $CC$-segment (see Figure~\ref{fig:2InsideR-3-1}), then case analysis shows that the two invariants are preserved and $\chi$ decreases. 

\begin{figure}[htb]
 \centering
 \phantom{.}\hfill\includegraphics[scale=\graphicsScale,page=2]{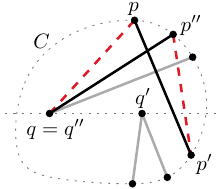}\hfill\hfill
 \includegraphics[scale=\graphicsScale,page=3]{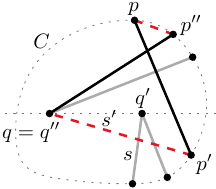}\hfill\hfill
 \includegraphics[scale=\graphicsScale,page=5]{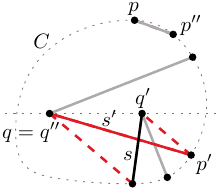}\hfill\hfill
 \includegraphics[scale=\graphicsScale,page=4]{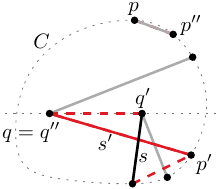}\hfill\phantom{.}\\
 \phantom{.}\hfill (a) \hfill\hfill (b) \hfill\hfill (c) \hfill\hfill (d) \hfill\phantom{.}
 \caption{Theorem~\ref{thm:2InsideR}, Phase~3 when $pp'$ is a peripheral segment.}
 \label{fig:2InsideR-3-2}
\end{figure}

If $pp'$ is a peripheral $CC$-segment (see Figure~\ref{fig:2InsideR-3-2}), then a case analysis shows that the two invariants are preserved and $\chi$ has the following behavior. If no $CC$-segment is inserted, then $\chi$ decreases. Otherwise a $CC$-segment and a $TT$-segment are inserted and $\chi$ may increase by $\OO(t)$ (Figure~\ref{fig:2InsideR-3-2}(d)). Notice that the number of times the $TT$-segment $qq'$ is inserted is $\OO(t)$, which bounds the total increase by $\OO(t^2)$. 

As $\chi = \OO(tn)$, the total increase is $\OO(t^2)$, and $\chi$ decreases at all but $\OO(t)$ steps, we have that the number of flips in Phase~3 is $\OO(tn)$.

\textbf{Phase~4 ($\mathbf{CT \times \textbf{central } CC}$).} 
At this point, each crossing involves a central $CC$-segment and either a $CT$-segment or the $TT$-segment $qq'$.
In this phase, we remove all crossings between $CT$-segments and central $CC$-segments, ignoring the $TT$-segments.
This phase ends with crossings only between $qq'$ and central $CC$-segments.

\begin{figure}[htb]
 \centering
 \phantom{.}\hfill\includegraphics[scale=\graphicsScale,page=1]{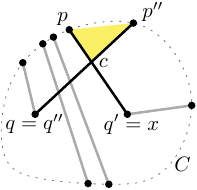}\hfill\hfill
 \includegraphics[scale=\graphicsScale,page=2]{2InsideRphase4}\hfill\hfill
 \includegraphics[scale=\graphicsScale,page=4]{2InsideRphase4}\hfill\hfill
 \includegraphics[scale=\graphicsScale,page=3]{2InsideRphase4}\hfill\phantom{.}\\
 \phantom{.}\hfill (a) \hfill\hfill (b) \hfill\hfill (c) \hfill\hfill (d) \hfill\phantom{.}
 \caption{Theorem~\ref{thm:2InsideR}, Phase~4. (a) A pair of $CT$-segments with an ear. (b) A $CC$-segment and a $CT$-segment with an ear. (c) Flipping an ear that produces crossing pairs of $CT$-segments. (d) Flipping an ear that inserts a non-central $CC$-segment with crossings.}
 \label{fig:2InsideR-4}
\end{figure}

Given four endpoints $q'' \in T$, $p,p'' \in C$, and $x \in C \cup T$, we say that a pair of segments $p''q'',xp \in S$ crossing at a point $c$ contains an \emph{ear} $\widehat{pp''}$ if the interior of the triangle $pp''c$ intersects no segment of $S$ (see Figure~\ref{fig:2InsideR-4}(a) and~\ref{fig:2InsideR-4}(b)).
Every set of segments with endpoints in $C \cup T$ with $|T| = 2$ that has crossings (not involving the $TT$-segment) contains an ear (adjacent to the crossing that is farthest from the line $qq'$).

At each \emph{step}, we flip a pair of segments $p''q'',xp$ that contains an ear $\widehat{pp''}$, prioritizing pairs where both segments are $CT$-segments. Notice that, even though initially we did not have crossing pairs of $CT$-segments, they may be produced in the flip (Figure~\ref{fig:2InsideR-4}(c)).
If the flip inserts a non-central $CC$-segment which crosses some $CT$-segments (Figure~\ref{fig:2InsideR-4}(d)), then, we perform the following \emph{while loop}. Assume without loss of generality that $qq'$ is horizontal and $s$ is closer to $q'$ than to $q$. While there exists a non-central $CC$-segment $s$ with crossings, we flip $s$ with the $CT$-segment $s'$ crossing $s$ that comes first according to the following order. As a first criterion, a segment incident to $q$ comes before a segment incident to $q'$. As a second tie-breaking criterion, a segment whose crossing point with $s$ that is farther from the line $qq'$ comes before one that is closer.

Let $\chi = \OO(tn)$ be the number of crossings between central $CC$-segments and $CT$-segments plus the number of crossings between $CT$-segments.
A case analysis shows that the value of $\chi$ decreases at each step. If no non-central $CC$-segment is inserted, then the corresponding step consists of a single flip. As $\chi$ decreases, there are $\OO(tn)$ steps that do not insert a non-central $CC$-segment.

However, if a non-central $CC$-segment is inserted, at the end of the step we inserted a $CC$-segment that can no longer be flipped (Lemma~\ref{lem:split}). As the number of $CC$-segments is $\OO(n)$, we have that the number of times the while loop is executed is $\OO(n)$. Since each execution of the while loop performs $\OO(t)$ flips, we have a total of $\OO(tn)$ flips in this phase.

\textbf{Phase~5 ($\mathbf{TT \times \textbf{central } CC}$).} In this phase, we remove all crossings left, which are between the possibly multiple copies of the $TT$-segment $qq'$ and central $CC$-segments. The endpoints of the segments with crossings are in convex position and all other endpoints are outside their convex hull. Hence, by Lemma~\ref{lem:split}, it is possible to obtain a crossing-free multigraph using $\OO(d_{\text{conv}}(n))$ flips.
\end{proof}

\subsection{One Point inside and One Point Outside a Convex}
\label{sec:1Inside1OutsideR}

Given an endpoint $p$, let $\delta(p)$ denote the degree of $p$, that is, the number of segments incident to $p$. The following lemma is used to prove Theorem~\ref{thm:1Inside1OutsideR}.

\begin{lemma} \label{lem:1Inside1OutsideR}
Consider a multigraph $(P,S)$ with $P = C \cup T$ where $C$ is in convex position, and $T = \{q,q'\}$ such that $q$ is outside the convex hull of $C$ and $q'$ is inside the convex hull of $C$. Consider that $q$ is the endpoint of a single segment $s$ and all crossings are on $s$. 
Let $n = |S|$ and $t = \OO(\delta(q'))$ be the number of $T$-segments.
There exists a flip sequence of length $\OO(tn)$ using only removal choice that ends with all crossings (if any) on the segment $qq'$.

\end{lemma}
\begin{proof}
We proceed as follows, while $s$ has crossings. For induction purpose, let $f(n')$ be the length of the flip sequence in the lemma statement for $n' < n$ segments.

Let $s'$ be the segment that crosses $s$ at the point farthest from $q$. We flip $s$ and $s'$, arriving at one of the three cases below (Figure~\ref{fig:1Inside1OutsideR}).

\begin{figure}[!ht]
  \centering\hfill
  \includegraphics[page=1,scale=\graphicsScale]{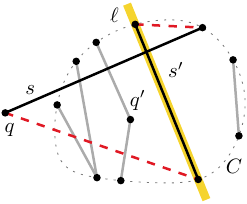}\hfill\hfill
  \includegraphics[page=2,scale=\graphicsScale]{1Inside1OutsideR}\hfill\hfill
  \includegraphics[page=3,scale=\graphicsScale]{1Inside1OutsideR}\hfill~\\
  \hfill Case~1\hfill\hfill Case~2\hfill\hfill Case~3\hfill~
    \caption{The three cases in the proof of Lemma~\ref{lem:1Inside1OutsideR}.}
    \label{fig:1Inside1OutsideR}
\end{figure}

\para{Case~1 ($\mathbf{CT \times CC)}$.} In this case, the segment $s'$ is a $CC$-segment. Notice that the line $\ell$ containing $s'$ becomes crossing free after the flip. There are segments on both sides of $\ell$. 
If $\ell$ separates $q,q'$, then we untangle both sides independently (Lemma~\ref{lem:split}) using $\OO(n)$ and $\OO(t n)$ flips (Theorem~\ref{thm:1InsideOutsideR}). Otherwise, the segments on one side of $\ell$ are already crossing free (because of the specific choice of $s'$) and we inductively untangle the $n' \leq n-1$ segments on the other side of $\ell$ using $f(n')$ flips.

\para{Case~2 ($\mathbf{CT \times CT \rightarrow CC,TT}$).} If $s'$ is a $CT$-segment and one of the inserted segments is the $TT$-segment $qq'$, then the procedure is over as all crossings are on $qq'$.

\para{Case~3 ($\mathbf{CT \times CT \rightarrow CT,CT}$).} In this case two $CT$-segments are inserted. Let $p \in C$ be an endpoint of $s = qp$. Since the inserted $CT$-segment $q'p$ is crossing free, Case~3 only happens $\OO(t)$ times before we arrive at Case~1 or Case~2.

Putting the three cases together, we obtain the recurrence
\[f(n) \leq \OO(t) + f(n')\text{, with } n' \leq n-1,\]
which solves to $f(n) = \OO(tn)$, as claimed.
\end{proof}

We are now ready to prove the theorem.

\begin{theorem} \label{thm:1Inside1OutsideR}
Consider a multigraph $(P,S)$ with $P = C \cup T$ where $C$ is in convex position, and $T = \{q,q'\}$ such that $q$ is outside the convex hull of $C$ and $q'$ is inside the convex hull of $C$. 
Let $n = |S|$ and $t$ be the number of $T$-segments.
There exists an untangle sequence of length $\OO(d_{\text{conv}}(n) + \delta(q)\delta(q')n) = \OO(d_{\text{conv}}(n) + t^2n)$ using only removal choice, where $d_{\text{conv}}(n)$ is the number of flips to untangle any multiset of at most $n$ segments with endpoints in convex position.
\end{theorem}
\begin{proof}
The untangle sequence contains four phases.

\textbf{Phase~1 ($\mathbf{CC\times CC}$).} In this phase, we remove all crossings between pairs of $CC$-segments using $d_{\text{conv}}(n)$ flips. Throughout all the phases, the invariant that no pair of $CC$-segments crosses is preserved.

\textbf{Phase~2 ($\mathbf{Cq' \times CC}$).} In this phase, we remove all crossings between pairs composed of a $CC$-segment and a $CT$-segment incident to $q'$ (the point inside the convex hull of $C$) using $\OO(tn)$ flips by Theorem~\ref{thm:1InsideOutsideR}.

\textbf{Phase~3 ($\mathbf{Cq}$).}
At this point, all crossings involve a segment incident to $q$. In this phase, we deal with all remaining crossings except the crossings involving the segment $qq'$. Lemma~\ref{lem:split} allows us to remove the crossings in each $CT$-segment $s$ incident to $q$ independently, which we do using $\OO(\delta(q') n)$ flips using Lemma~\ref{lem:1Inside1OutsideR}.
As there are $\delta(q)$ $CT$-segments adjacent to $q$, the total number of flips is $\OO(\delta(q) \delta(q') n) = \OO(t^2n)$.

\textbf{Phase~4 ($\mathbf{CC \times TT}$).} At this point, all crossings involve the $TT$-segment $qq'$. The endpoints in $C$ that are adjacent to segments with crossings, together with $q'$, are all in convex position. Hence, the only endpoint not in convex position is $q$, and we apply Theorem~\ref{thm:1InsideOutsideR} using $\OO(tn)$ flips.

After the $d_{\text{conv}}(n)$ flips in Phase~1, the number of flips is dominated by Phase~3 with $\OO(\delta(q) \delta(q') n) = \OO(t^2n)$ flips.
\end{proof}

Notice that, in certain cases (for example in the red-blue case with $q,q'$ having different colors) a flip between two $CT$-segments never produces two $CT$-segments. Consequently, Case~3 of the proof of Lemma~\ref{lem:1Inside1OutsideR} never happens, and the bound in Theorem~\ref{thm:1Inside1OutsideR} decreases to $\OO(d_{\text{conv}}(n) + tn)$.

\section{Removal and Insertion Choices}
\label{sec:RemovalInsertion}

In this section, we show how to untangle a matching or a multigraph using both removal and insertion choices. We start with the case of points outside the convex separated by two parallel lines. Afterwards, we prove an important lemma and apply it to untangle a matching with points outside the convex.

\subsection{Separated by Two Parallel Lines}
\label{sec:separated}

We start with the simpler case in which $T$ is separated from $C$ by two parallel lines. In this case, our bound of $\OO(n + t|P|)$ interpolates the tight convex bound of $\OO(n)$ from~\cite{BMS19} and the $\OO(t|P|)$ bound from~\cite{BoM16} for $t$ arbitrary segments.

\begin{theorem} \label{thm:separatedRI}
Consider a multigraph $(P,S)$ with $P = C \cup T_1 \cup T_2$ where $C$ is in convex position and there exist two horizontal lines $\ell_1,\ell_2$, with $T_1$ above $\ell_1$ above $C$ above $\ell_2$ above $T_2$.
Let $n = |S|$, $T = T_1 \cup T_2$, and $t$ be the number of $T$-segments.
There exists an untangle sequence of length $\OO(n + t|P|) = \OO(tn)$ using both removal and insertion choices.
\end{theorem}
\begin{proof}
The algorithm runs in two phases.

\para{Phase~1.} We use removal choice to perform the flips involving a point in $T$. At the end of the first phase, there can only be crossings among segments with all endpoints in $C$. 
The insertion choice for the first phase is the following.
Let $p_1,\ldots,p_{|P|}$ be the points $P$ sorted vertically from top to bottom.
Consider a flip involving the points $p_a,p_b,p_c,p_d$ with $a<b<c<d$. The insertion choice is to create the segments $p_ap_b$ and $p_cp_d$. As in~\cite{BoM16}, we define the potential $\eta$ of a segment $p_ip_j$ as
$\eta(p_ip_j) = |i-j|$.
Notice that $\eta$ is an integer from $1$ to $|P|-1$. We define $\eta(S)$ as the sum of $\eta(p_ip_j)$ for $p_ip_j \in S$ with $p_i$ or $p_j$ in $T$. Notice that $0 < \eta(S) < t |P|$. It is easy to verify that any flip involving a point in $T$ decreases $\eta(S)$. Hence, the number of flips in Phase~1 is $\OO(t|P|)$.

\para{Phase~2.}
Since $T$ is outside the convex hull of $C$, flips between segments with all endpoints in $C$ cannot create crossings with the other segments, which are guaranteed to be crossing free at this point. Hence, it suffices to run an algorithm to untangle a convex set with removal and insertion choice from~\cite{BMS19}, which performs $\OO(n)$ flips.
\end{proof}

\subsection{Liberating a Line}
\label{sec:libLine}

In this section, we prove the following key lemma, which we use in the following section. The lemma only applies to matchings and it is easy to find a counter-example for multisets ($S$ consisting of $n$ copies of a single segment that crosses $pq$).

\begin{lemma}
    \label{lem:libLine}
    Consider a matching $S$ of $n$ segments with endpoints $C$ in convex position, and a segment $pq$ separating $C$. Using $\OO(n)$ flips with removal and insertion choices on the initial set $S \cup \{pq\}$, we obtain a set of segments that do not cross the line $pq$. 
\end{lemma}

\begin{proof}
    For each flip performed in the subroutine described hereafter, at least one of the inserted segments does not cross the line $pq$ and is removed from $S$ (see Figure~\ref{fig:libLine}). 
    
    \begin{figure}[htb]
     \centering
     \includegraphics[scale=\graphicsScale,page=1]{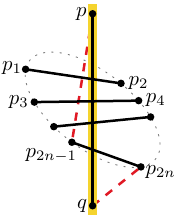} \hfill
     \includegraphics[scale=\graphicsScale,page=2]{libLine} \hfill
     \includegraphics[scale=\graphicsScale,page=3]{libLine} \hfill
     \includegraphics[scale=\graphicsScale,page=4]{libLine} \hfill
     \includegraphics[scale=\graphicsScale,page=5]{libLine}\\
     \caption{An untangle sequence of the subroutine to liberate the line $pq$ (with $n=4$).}
     \label{fig:libLine}
    \end{figure}
    
    \para{Preprocessing.}
    First, we remove from $S$ the segments that do not intersect the line $pq$, as they are irrelevant.
    Second, anytime two segments in $S$ cross, we flip them choosing to insert the pair of segments not crossing the line $pq$. One such flip removes two segments from $S$.
    Let $p_1p_2$ (respectively $p_{2n-1}p_{2n}$) be the segment in $S$ whose intersection point with $pq$ is the closest from $p$ (respectively $q$).
    Without loss of generality, assume that the points $p_1$ and $p_{2n-1}$ are on the same side of the line $pq$.
    
    \para{First flip.}
    Lemma~\ref{lem:triangleHide} applied to the segment $pq$ and the triangle $p_1p_2p_{2n-1}$ shows that at least one of the segments among $pp_{2n-1},qp_1,qp_2$ intersects all the segments of $S$.
    Without loss of generality, assume that $pp_{2n-1}$ is such a segment, i.e., that $pp_{2n-1}$ crosses all segments of $S \setminus \{p_{2n-1}p_{2n}\}$. 
    We choose to remove the segments $pq$ and $p_{2n-1}p_{2n}$, and we choose to insert the segments $pp_{2n-1}$ and $qp_{2n}$.
    As the segment $qp_{2n}$ does not cross the line $pq$, we remove it from $S$.
    
    \para{Second flip.} We choose to flip the segments $pp_{2n-1}$ and $p_1p_2$.
    If $n$ is odd, we choose to insert the pair of segments $pp_1,p_2p_{2n-1}$.
    If $n$ is even, we insert the segments $pp_2,p_1p_{2n-1}$. 
    
    By convexity, one of the inserted segment (the one with endpoints in $C$) crosses all other $n-2$ segments.
    The other inserted segment (the one with $p$ as one of its endpoints) does not cross the line $pq$, so we remove it from $S$.
    Note that the condition on the parity of $n$ is there only to ensure that the last segment $p_{2n-3}p_{2n-2}$ is dealt with at the last flip.
    
    \para{Remaining flips.} 
    We describe the third flip. The remaining flips are performed similarly.
    Let $s$ be the previously inserted segment.
    Let $p_3p_4$ be the segment in $S$ whose intersection point with $pq$ is the closest from $p$. Without loss of generality, assume that $p_3$ is on the same side of the line $pq$ as $p_1$ and $p_{2n-1}$. 
    
    We choose to flip $s$ with $p_3p_4$.
    If $s = p_2p_{2n-1}$, we choose to insert the pair of segments $p_2p_4,p_3p_{2n-1}$.
    If $s = p_1p_{2n-1}$, we choose to insert the pair of segments $p_1p_3,p_4p_{2n-1}$.
    
    By convexity, one inserted segment (the one with $p_{2n-1}$ as an endpoint) crosses all other $n-3$ segments. 
    The other inserted segment does not cross the line $pq$, so we remove it from $S$.
    Note that the insertion choice described is the only viable one, as the alternative would insert a crossing-free segment crossing the line $pq$ that cannot be removed.
\end{proof}

\subsection{Points Outside a Convex}
\label{sec:outside}

We are now ready to prove the following theorem, which only applies to matchings because it uses Lemma~\ref{lem:libLine}.

\begin{theorem} \label{thm:nearConvexRI}
Consider a matching $S$ consisting of $n$ segments with endpoints $P = C \cup T$ where $C$ is in convex position and $T$ is outside the convex hull of $C$.
Let $t = |T|$.
There exists an untangle sequence of length $\OO(t^3n)$ using both removal and insertion choices.
\end{theorem}
\begin{proof}
Throughout this proof, we partition the $TT$-segments into two types: \emph{$TTI$-segment} if it intersects the interior of the convex hull of $C$ and \emph{$TTO$-segment} otherwise.

\para{$\mathbf{TT}$-segments.}
At any time during the untangle procedure, if there is a $TTI$-segment $s$ that crosses more than $t$ segments, we apply Lemma~\ref{lem:libLine} to liberate $s$ from every $CC$-segment using $\OO(n)$ flips.
Let $\ell$ be the line containing $s$. Since $\lambda(\ell)$ cannot increase (Lemma~\ref{lem:lambda}), $\lambda(\ell) < t$ after Lemma~\ref{lem:libLine}, and there are $\OO(t^2)$ different $TTI$-segments, it follows that Lemma~\ref{lem:libLine} is applied $\OO(t^2)$ times, performing a total $\OO(t^2n)$ flips.
As the number of times $s$ is inserted and removed differ by at most $1$ and $\lambda(\ell)$ decreases at each flip that removes $s$, it follows that $s$ participates in $\OO(t)$ flips. As there are $\OO(t^2)$ different $TTI$-segments, the total number of flips involving $TTI$-segments is $\OO(t^3)$.

We define a set $L$ of $\OO(t)$ lines as follows. For each point $q \in T$, we have two lines $\ell_1, \ell_2 \in L$ that are the two tangents of the convex hull of $C$ that pass through $q$. As the lines $\ell \in L$ do not separate $C$, the potential $\lambda(\ell) = \OO(t)$.
When flipping a $TTO$-segment $q_1q_2$ with another segment $q_3p$ with $q_3 \in T$ ($p$ may be in $T$ or in $C$), we make the insertion choice of creating a $TTO$-segment $q_1q_3$ such that there exists a line $\ell \in L$ whose potential $\lambda(\ell)$ decreases. It is easy to verify that $\ell$ always exist (see Lemmas~\ref{lem:Icritical} and~\ref{lem:criticalTangent} in the Appendix). Hence, the number of flips involving $TTO$-segments is $\OO(t^2)$ and the number of flips involving $TT$-segments in general is $\OO(t^3)$.

\para{All except pairs of $\mathbf{CC}$-segments.}
We keep flipping segments that are not both $CC$-segments with the following insertion choices.
Whenever we flip two $CT$-segments, we make the insertion choice of creating a $TT$-segment. Hence, as the number of flips involving $TT$-segments is $\OO(t^3)$, so is the number of flips of two $CT$-segments.

Whenever we flip a $CT$-segment $p_1q$ with $q \in T$ and a $CC$-segment $p_3p4$, we make the following insertion choice. Let $v(q)$ be a vector such that the dot product $v(q) \cdot q < v(q) \cdot p$ for all $p \in C$, that is, $v$ is orthogonal to a line $\ell$ separating $q$ from $C$ and pointing towards $C$. We define the potential $\eta(p_xq)$ of a segment with $p_x \in C$ and $q \in T$ as the number of points $p \in C$ such that $v(q) \cdot p < v(q) \cdot p_x$, that is the number of points in $C$ before $p_x$ in direction $v$. We choose to insert the segment $p_xq$ that minimizes $\eta(p_xq)$ for $x = \{1,2\}$. Let $\eta(S)$ be the sum of $\eta(p_xq)$ for all $CT$-segments $p_xq$ in $S$. It is easy to see that $\eta(S)$ is $\OO(t|C|)$ and decreases at each flip involving a $CT$-segment (not counting the flips inside Lemma~\ref{lem:libLine}).

There are two situation in which $\eta(S)$ may increase. One is when Lemma~\ref{lem:libLine} is applied, which happens $\OO(t^2)$ times. Another one is when a $TT$-segment and a $CC$-segment flip, creating two $CT$-segments, which happens $\OO(t^3)$ times. At each of these two situations, $\eta(S)$ increases by $\OO(|C|)$. Consequently, the number of flips between a $CT$-segment and a $CC$-segment is $\OO(t^3|C|) = \OO(t^3n)$.

\para{$\mathbf{CC}$-segments.}
By removal choice, we choose to flip the pairs of $CC$-segments last (except for the ones flipped in Lemma~\ref{lem:libLine}). As $T$ is outside the convex hull of $C$, flipping two $CC$-segments does not create crossings with other segments (Lemma~\ref{lem:split}). Hence, we apply the algorithm from~\cite{BMS19} to untangle the remaining segments using $\OO(n)$ flips.
\end{proof}

\bibliography{ref}

\begin{thebibliography}{10}

\bibitem{ABDK22}
Oswin Aichholzer, Brad Ballinger, Therese Biedl, Mirela Damian, Erik~D Demaine,
  Matias Korman, Anna Lubiw, Jayson Lynch, Josef Tkadlec, and Yushi Uno.
\newblock Reconfiguration of non-crossing spanning trees.
\newblock {\em arXiv preprint}, 2022.
\newblock URL: \url{https://arxiv.org/abs/2206.03879}.

\bibitem{AKLM23}
Oswin Aichholzer, Kristin Knorr, Maarten L{\"o}ffler, Zuzana Mas{\'a}rov{\'a},
  Wolfgang Mulzer, Johannes Obenaus, Rosna Paul, and Birgit Vogtenhuber.
\newblock Flipping plane spanning paths.
\newblock In {\em International Conference and Workshops on Algorithms and
  Computation (WALCOM)}, 2023.
\newblock URL: \url{https://arxiv.org/abs/2202.10831}, \href
  {https://doi.org/10.1007/978-3-031-27051-2_5}
  {\path{doi:10.1007/978-3-031-27051-2_5}}.

\bibitem{AMP15}
Oswin Aichholzer, Wolfgang Mulzer, and Alexander Pilz.
\newblock Flip distance between triangulations of a simple polygon is
  {NP}-complete.
\newblock {\em Discrete \& Computational Geometry}, 54(2):368--389, 2015.
\newblock \href {https://doi.org/10.1007/s00454-015-9709-7}
  {\path{doi:10.1007/s00454-015-9709-7}}.

\bibitem{AIM07}
Selim~G Akl, Md~Kamrul Islam, and Henk Meijer.
\newblock On planar path transformation.
\newblock {\em Information processing letters}, 104(2):59--64, 2007.
\newblock \href {https://doi.org/10.1016/j.ipl.2007.05.009}
  {\path{doi:10.1016/j.ipl.2007.05.009}}.

\bibitem{BeI08}
Sergey Bereg and Hiro Ito.
\newblock Transforming graphs with the same degree sequence.
\newblock In {\em Computational Geometry and Graph Theory}, pages 25--32, 2008.
\newblock \href {https://doi.org/10.1007/978-3-540-89550-3_3}
  {\path{doi:10.1007/978-3-540-89550-3_3}}.

\bibitem{BeI17}
Sergey Bereg and Hiro Ito.
\newblock Transforming graphs with the same graphic sequence.
\newblock {\em Journal of Information Processing}, 25:627--633, 2017.
\newblock \href {https://doi.org/10.2197/ipsjjip.25.627}
  {\path{doi:10.2197/ipsjjip.25.627}}.

\bibitem{BMS19}
Ahmad Biniaz, Anil Maheshwari, and Michiel Smid.
\newblock Flip distance to some plane configurations.
\newblock {\em Computational Geometry}, 81:12--21, 2019.
\newblock URL: \url{https://arxiv.org/abs/1905.00791}, \href
  {https://doi.org/10.1016/j.comgeo.2019.01.008}
  {\path{doi:10.1016/j.comgeo.2019.01.008}}.

\bibitem{BBH19}
Marthe Bonamy, Nicolas Bousquet, Marc Heinrich, Takehiro Ito, Yusuke Kobayashi,
  Arnaud Mary, Moritz M{\"{u}}hlenthaler, and Kunihiro Wasa.
\newblock The perfect matching reconfiguration problem.
\newblock In {\em 44th International Symposium on Mathematical Foundations of
  Computer Science (MFCS)}, volume 138 of {\em LIPIcs}, pages 80:1--80:14,
  2019.
\newblock \href {https://doi.org/10.4230/LIPIcs.MFCS.2019.80}
  {\path{doi:10.4230/LIPIcs.MFCS.2019.80}}.

\bibitem{BoM16}
{\'{E}}douard Bonnet and Tillmann Miltzow.
\newblock Flip distance to a non-crossing perfect matching.
\newblock {\em arXiv}, 1601.05989, 2016.
\newblock URL: \url{http://arxiv.org/abs/1601.05989}.

\bibitem{BJ20}
Nicolas Bousquet and Alice Joffard.
\newblock Approximating shortest connected graph transformation for trees.
\newblock In {\em Theory and Practice of Computer Science}, pages 76--87, 2020.
\newblock \href {https://doi.org/10.1007/978-3-030-38919-2_7}
  {\path{doi:10.1007/978-3-030-38919-2_7}}.

\bibitem{ChWu09}
Jou-Ming Chang and Ro-Yu Wu.
\newblock On the diameter of geometric path graphs of points in convex
  position.
\newblock {\em Information processing letters}, 109(8):409--413, 2009.
\newblock \href {https://doi.org/10.1016/j.ipl.2008.12.017}
  {\path{doi:10.1016/j.ipl.2008.12.017}}.

\bibitem{DGR23}
Guilherme~D. da~Fonseca, Yan Gerard, and Bastien Rivier.
\newblock On the longest flip sequence to untangle segments in the plane.
\newblock In {\em International Conference and Workshops on Algorithms and
  Computation (WALCOM)}, volume to appear of {\em Lecture Notes in Computer
  Science}, 2023.
\newblock URL: \url{https://arxiv.org/abs/2210.12036}, \href
  {https://doi.org/10.1007/978-3-031-27051-2_10}
  {\path{doi:10.1007/978-3-031-27051-2_10}}.

\bibitem{DDFGR22}
Arun~Kumar Das, Sandip Das, Guilherme~D. da~Fonseca, Yan Gerard, and Bastien
  Rivier.
\newblock Complexity results on untangling red-blue matchings.
\newblock {\em Computational Geometry}, 111:101974, 2023.
\newblock URL: \url{https://arxiv.org/abs/2202.11857}, \href
  {https://doi.org/10.1016/j.comgeo.2022.101974}
  {\path{doi:10.1016/j.comgeo.2022.101974}}.

\bibitem{ERV14}
Matthias Englert, Heiko R{\"o}glin, and Berthold V{\"o}cking.
\newblock Worst case and probabilistic analysis of the {2-Opt} algorithm for
  the {TSP}.
\newblock {\em Algorithmica}, 68(1):190--264, 2014.
\newblock \href {https://doi.org/10.1007/s00453-013-9801-4}
  {\path{doi:10.1007/s00453-013-9801-4}}.

\bibitem{EKM13}
P{\'e}ter~L Erd{\H{o}}s, Zolt{\'a}n Kir{\'a}ly, and Istv{\'a}n Mikl{\'o}s.
\newblock On the swap-distances of different realizations of a graphical degree
  sequence.
\newblock {\em Combinatorics, Probability and Computing}, 22(3):366--383, 2013.
\newblock \href {https://doi.org/10.1017/S0963548313000096}
  {\path{doi:10.1017/S0963548313000096}}.

\bibitem{Gru73}
Branko Grünbaum.
\newblock Polygons in arrangements generated by n points.
\newblock {\em Mathematics Magazine}, 46(3):113--119, 1973.
\newblock \href {https://doi.org/10.1080/0025570X.1973.11976293}
  {\path{doi:10.1080/0025570X.1973.11976293}}.

\bibitem{Hak62}
Seifollah~Louis Hakimi.
\newblock On realizability of a set of integers as degrees of the vertices of a
  linear graph. {I}.
\newblock {\em Journal of the Society for Industrial and Applied Mathematics},
  10(3):496--506, 1962.

\bibitem{Hak63}
Seifollah~Louis Hakimi.
\newblock On realizability of a set of integers as degrees of the vertices of a
  linear graph {II}. uniqueness.
\newblock {\em Journal of the Society for Industrial and Applied Mathematics},
  11(1):135--147, 1963.

\bibitem{HNU99}
Ferran Hurtado, Marc Noy, and Jorge Urrutia.
\newblock Flipping edges in triangulations.
\newblock {\em Discrete \& Computational Geometry}, 22(3):333--346, 1999.

\bibitem{Joff20}
Alice Joffard.
\newblock {\em Graph domination and reconfiguration problems}.
\newblock PhD thesis, Université Claude Bernard Lyon 1, 2020.

\bibitem{Law72}
Charles~L Lawson.
\newblock Transforming triangulations.
\newblock {\em Discrete Mathematics}, 3(4):365--372, 1972.

\bibitem{LuP15}
Anna Lubiw and Vinayak Pathak.
\newblock Flip distance between two triangulations of a point set is
  {NP}-complete.
\newblock {\em Computational Geometry}, 49:17--23, 2015.
\newblock \href {https://doi.org/10.1016/j.comgeo.2014.11.001}
  {\path{doi:10.1016/j.comgeo.2014.11.001}}.

\bibitem{NiN18}
Naomi Nishimura.
\newblock Introduction to reconfiguration.
\newblock {\em Algorithms}, 11(4), 2018.
\newblock \href {https://doi.org/10.3390/a11040052}
  {\path{doi:10.3390/a11040052}}.

\bibitem{OdW07}
Yoshiaki Oda and Mamoru Watanabe.
\newblock The number of flips required to obtain non-crossing convex cycles.
\newblock In {\em Kyoto International Conference on Computational Geometry and
  Graph Theory}, pages 155--165, 2007.

\bibitem{Pil14}
Alexander Pilz.
\newblock Flip distance between triangulations of a planar point set is
  {APX}-hard.
\newblock {\em Computational Geometry}, 47(5):589--604, 2014.
\newblock \href {https://doi.org/10.1016/j.comgeo.2014.01.001}
  {\path{doi:10.1016/j.comgeo.2014.01.001}}.

\bibitem{Heu13}
Jan van~den Heuvel.
\newblock The complexity of change.
\newblock {\em Surveys in Combinatorics}, 409:127--160, 2013.

\bibitem{VLSC81}
Jan van Leeuwen and Anneke~A. Schoone.
\newblock Untangling a traveling salesman tour in the plane.
\newblock In {\em 7th Workshop on Graph-Theoretic Concepts in Computer
  Science}, 1981.

\bibitem{Wil99}
Todd~G Will.
\newblock Switching distance between graphs with the same degrees.
\newblock {\em SIAM Journal on Discrete Mathematics}, 12(3):298--306, 1999.
\newblock \href {https://doi.org/10.1137/S0895480197331156}
  {\path{doi:10.1137/S0895480197331156}}.

\bibitem{WCL09}
Ro{-}Yu Wu, Jou{-}Ming Chang, and Jia{-}Huei Lin.
\newblock On the maximum switching number to obtain non-crossing convex cycles.
\newblock In {\em 26th Workshop on Combinatorial Mathematics and Computation
  Theory}, pages 266--273, 2009.

\end{thebibliography}

\appendix

\section{Auxiliary Lemma of Section~\ref{sec:libLine}}
\label{sec:libLineLemmas}

In this section, we prove Lemma~\ref{lem:triangleHide} used in the proof of Lemma~\ref{lem:libLine}.

Recall that, in the proof of Lemma~\ref{lem:libLine}, we have a convex quadrilateral $p_1p_2p_{2n}p_{2n-1}$ and a segment $pq$ crossing the segments $p_1p_2$ and $p_{2n}p_{2n-1}$ in this order when drawn from $p$ to $q$, and we invoke Lemma~\ref{lem:triangleHide} to show that at least one of the segments among $pp_{2n-1},qp_1,qp_2$ intersects all the segments of $S$.
Before proving Lemma~\ref{lem:triangleHide}, we detail how to apply it to this context.

Lemma~\ref{lem:triangleHide} applied to the segment $pq$ and the triangle $p_1p_2p_{2n-1}$ asserts that at least one of the following pairs of segments cross: $pp_{2n-1},p_1p_2$, or $qp_1,p_2p_{2n-1}$, or $qp_2,p_1p_{2n-1}$.
If the segments $pp_{2n-1},p_1p_2$ cross, then we are done.
If the segments $qp_1,p_2p_{2n-1}$ cross, then the segments $qp_1,p_{2n}p_{2n-1}$ also cross and we are done.
If the segments $qp_2,p_1p_{2n-1}$ cross, then the segments $qp_2,p_{2n}p_{2n-1}$ also cross and we are done.

Next, we state and prove Lemma~\ref{lem:triangleHide}.

\begin{lemma}
    \label{lem:triangleHide}
    For any triangle $abc$, for any segment $pq$ intersecting the interior of the triangle $abc$, there exists a segment $s \in \{pa,pb,pc,qa,qb,qc\}$ that intersects the interior of the triangle $abc$.
\end{lemma}
\begin{proof}
    If all $a,b,c,p,q$ are in convex position, then $p$ and the point among $a,b,c$ that is not adjacent to $p$ on the convex hull boundary define the segment $s$. Otherwise, since $p,q$ are not adjacent on the convex hull boundary, assume without loss of generality that $a$ is not a convex hull vertex and $p,b,q,c$ are the convex hull vertices in order. Then, either $ap$ or $aq$ intersects $bc$.
\end{proof}

\section{Auxiliary Lemmas of Section~\ref{sec:outside}}
\label{sec:outsideLemmas}

In this section, we prove Lemma~\ref{lem:criticalTangent} and Lemma~\ref{lem:Icritical} used in the proof of Theorem~\ref{thm:nearConvexRI}.

Recall that, in the proof of Theorem~\ref{thm:nearConvexRI}, we define a set $L$ of lines as follows.
For each point $q \in T$, we have two lines $\ell_1, \ell_2 \in L$ that are the two tangents of the convex hull of $C$ that pass through $q$.
When flipping a $TTO$-segment $q_1q_2$ with another segment $q_3p$ with $q_3 \in T$ ($p$ may be in $T$ or in $C$), we make the insertion choice of creating a $TTO$-segment $q_1q_3$ such that there exists a line $\ell \in L$ whose potential $\lambda(\ell)$ decreases.
We invoke Lemma~\ref{lem:Icritical} and Lemma~\ref{lem:criticalTangent} to show that such a line $\ell$ always exist.

Indeed, by Lemma~\ref{lem:Icritical}, it is enough to show that there exists a line $\ell \in L$ containing one of the points $q_1,q_2,q_3$ that crosses one of the segments $q_1q_2$ or $q_3p$. This is precisely what Lemma~\ref{lem:criticalTangent} shows.

Next, we state prove Lemma~\ref{lem:Icritical} and Lemma~\ref{lem:criticalTangent}.

\begin{lemma}
    \label{lem:Icritical}
    Consider two crossing segments $p_1p_2,p_3p_4$ and a line $\ell$ containing $p_1$ and crossing $p_3p_4$.
    Then, one of the two pairs of segments $p_1p_3,p_2p_4$ or $p_1p_4,p_2p_3$ does not cross $\ell$.
    In other words, there exists an insertion choice to flip $p_1p_2,p_3p_4$ such that the number of segments crossing $\ell$ decreases.
\end{lemma}
\begin{proof}
    Straightforward.
\end{proof}

\begin{lemma}
    \label{lem:criticalTangent}
    Consider a closed convex body $B$ and two crossing segments $q_1q_3,q_2q_4$ whose endpoints $q_1,q_2,q_3$ are not in $B$, and whose endpoint $q_4$ is not in the interior of $B$.
    If the segment $q_1q_3$ does not intersect the interior of $B$, then at least one of the six lines tangent to $B$ and containing one of the endpoints $q_1,q_2,q_3$ is crossing one of the segments $q_1q_3,q_2q_4$.
    (General position is assumed, meaning that the aforementioned six lines are distinct, i.e., each line does not contain two of the points $q_1,q_2,q_3,q_4$.)
\end{lemma}

\begin{figure}[!ht]
    \hspace*{\stretch{1}}
    \pbox[b]{\textwidth}{\centering\includegraphics[page=1,scale=\graphicsScale]{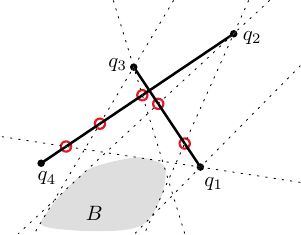}\newline(a)}\hspace*{\stretch{2}}
    \pbox[b]{\textwidth}{\centering\includegraphics[page=2,scale=\graphicsScale]{criticalTangent}\newline(b)}
    \hspace*{\stretch{1}}
    \caption{(a) In the statement of Lemma~\ref{lem:criticalTangent}, we assert the existence of points, circled in the figure, which are the intersection of a line tangent to $B$ and containing one of the points $q_1,q_2,q_3$. (b) In the proof of Lemma~\ref{lem:criticalTangent} by contraposition, we exhibit a point, circled in the figure, showing that $B$ intersects one of the segment $q_1q_3$.}
    \label{fig:criticalTangent}
\end{figure}

\begin{proof}
    For all $i \in \{1,2,3\}$, let $\ell_i$ and $\ell_i'$ be the two lines containing $q_i$ and tangent to $B$.
    By contraposition, we assume that none of the six lines $\ell_1,\ell_1',\ell_2,\ell_2',\ell_3,\ell_3'$ crosses one of the segments $q_1q_3,q_2q_4$. In other words, we assume that the six lines are tangent to the convex quadrilateral $q_1q_2q_3q_4$.
    It is well known that, if $m \geq 5$, then any arrangement of $m$ lines or more admits at most one face with $m$ edges (see~\cite{Gru73} for example).
    Therefore, $B$ is contained in the same face of the arrangement of the six lines as the quadrilateral $q_1q_2q_3q_4$.
    Let $p_1$ (respectively $p_1'$) be a contact point between the line $\ell_1$ (respectively $\ell_1'$) and the convex body $B$.
    The segment $p_1p_1'$ crosses the segment $q_1q_3$ and is contained in $B$ by convexity, concluding the proof by contraposition.
\end{proof}

\end{document}